\documentclass[11pt]{article}

\usepackage{amsfonts,amssymb,amsmath,latexsym,xfrac,ae,aecompl}

\usepackage{graphicx}

\newcommand{\F}{\vspace*{\smallskipamount}}

\newcommand{\FFF}{\vspace*{\bigskipamount}}
\newcommand{\B}{\vspace*{-\smallskipamount}}
\newcommand{\BB}{\vspace*{-\medskipamount}}
\newcommand{\BBB}{\vspace*{-\bigskipamount}}

\newcommand{\T}{\hspace*{2em}}
\newcommand{\TT}{\hspace*{4em}}
\newcommand{\TTT}{\hspace*{6em}}

\newcommand{\cA}{\mathcal{A}}

\newcommand{\cO}{\mathcal{O}}

\newcommand{\cW}{\mathcal{W}}

\newcommand{\mE}{\mathbb{E}}

\newcommand{\Paragraph}[1]{\BBB\paragraph{#1}}
\newcommand{\remove}[1]{}

\newlength{\pagewidth}
\newlength{\captionwidth}
\newcommand{\qed}{\hfill $\square$ \smallbreak}
\newenvironment{proof}{\noindent{\bf Proof:}}{\qed}

\newtheorem{theorem}{Theorem}
\newtheorem{lemma}{Lemma}

\newcommand{\Min}[2]{{\,\min\,[\,#1\, , \,#2\,]\,}}

\begin{document}

\thispagestyle{empty}

\baselineskip    3ex
\parskip         1ex

\title{Broadcasting Spanning Forests on a Multiple-Access Channel~\footnotemark[1]\FFF\FFF\FFF }

\author{Bogdan S. Chlebus~\footnotemark[2] 
\and
Karol Go\l{\k a}b~\footnotemark[3] 
\and
Dariusz R. Kowalski~\footnotemark[4]
}

\date{}

\maketitle

\footnotetext[1]{
This work was published as~\cite{ChlebusGK-TCS03}.
The results of this paper were presented in a preliminary form in~\cite{ChlebusGK-SIROCCO02}.}

\footnotetext[2]{
 	Department of Computer Science and Engineering,
 	University of Colorado Denver,
 	Denver, CO~80217, USA.
	Work supported by the National Science Foundation under Grant No. 0310503.}

\footnotetext[3]{
 	Instytut Informatyki,
 	Uniwersytet Warszawski,
 	Warszawa 02-097, 
	Poland.}

\footnotetext[4]{Department of Computer Science, 
University of Liverpool, 
Liverpool L69~3BX, United Kingdom.}

\vfill

\begin{abstract}
The problem of finding a spanning forest of a graph 
in a distributed-processing environment is studied.
If an input graph is weighted, then the goal is to find 
a minimum-weight spanning forest.
The processors communicate by broadcasting.
The output consists of the edges that make 
a spanning forest and have been broadcast on the network.
Input edges are distributed among the processors, with
each edge held by one processor.

The underlying broadcast network is implemented as a multiple-access 
channel.
If exactly one processor attempts to perform a broadcast,
then the broadcast is successful.
A message broadcast successfully is delivered to all the
processors in one step.
If more than one processors broadcast simultaneously, then the
messages interfere with each other and no processor can
receive any of them.

Optimality of algorithmic solutions is investigated,
by way of comparing deterministic with randomized algorithms, 
and adaptive with oblivious ones.
Lower bounds are proved that either justify the optimality 
of specific algorithms or show that the optimal performance 
depends on a class of algorithms. 
\end{abstract}

\vfill

\thispagestyle{empty}

\setcounter{page}{0}

\newpage

\section{Introduction}

We consider a distributed system in which processing units 
communicate by broadcasting. 
The underlying broadcast network is implemented 
as a multiple-access channel. 
The processing units are called stations.
A message sent at a step is successfully received by all the stations 
only if there is exactly one station that performs a broadcast
at this step.
We work with a model in which a collision, of the messages sent
simultaneously by different stations at a step, results in a feedback
that allows all the stations to detect the collision.

We study the problem of finding a spanning forest of a graph.
The input may be either a simple graph, with no weights assigned
to the edges, or it may have a weight given for each edge, then the
goal is to find a minimum-weight spanning forest.
The input is distributed among the stations: 
each edge is held by one station.
Some edges are heard on the channel during an execution of an algorithm, 
we require that {\em all\/} the edges of a spanning forest we seek
are revealed explicitly by having been heard at some step.

We consider two possible ways to obtain an input.
In a static case, the input is provided to the
participating stations at the start of an execution. 
In a dynamic case, there is an adversary that controls the timing 
of when the station holding a given edge 
wakes up and joins the process of computation.

The questions we address in this paper concern the time
efficiency of finding a spanning forest.
We compare adaptive and oblivious algorithms,
and deterministic with randomized ones, and prove lower bounds.
For a given class of algorithms, the goal is to identify the optimum time
performance of an algorithm in the class.
When comparing classes of algorithms, the question is whether 
the optimum performance differs among the classes.

Now we give a detailed overview of the results of this
paper.
Let~$n$ be the number of vertices and let~$m$ be the number of edges 
of an input graph, which also means that there are $m$~stations 
participating in the computation.
The size~$|T|$ of a spanning forest~$T$ is defined to be equal 
to the number of edges in~$T$. 
A deterministic algorithm is said to be {\em oblivious\/} if its 
actions are defined in advance for any input of a given size, 
otherwise an algorithm is {\em adaptive.}

\begin{enumerate}

\item[\bf I.] 
{\sf Simple graphs and deterministic algorithms.}
We give a deterministic {\em adaptive\/} 
algorithm that finds a spanning forest~$T$
in time $\cO(\Min{m}{|T|\log m})$. 
We prove that, for any number of edges~$m$ and
any {\em oblivious\/} deterministic algorithm~$\cA$, there exists 
some graph~$G_\cA$ with $\Theta(m)$~edges that makes~$\cA$ perform
$\Omega(m)$~steps on~$G_\cA$.
This shows that adaptive and oblivious algorithms have 
different optimum-time performances, for sufficiently many
edges in the output.

\item[\bf II.] 
{\sf Simple graphs and randomized algorithms.}
We give a randomized algorithm that finds a spanning
forest~$T$ in expected time proportional to its
size~$|T|$, which is optimal.

\item[\bf III.] 
{\sf Graphs with weights and deterministic algorithms.}
We prove that any deterministic adaptive algorithm requires
time $\Omega(m)$ to find a minimum-weight spanning forest.
This shows, if $m=\omega (n\log n)$, that deterministic 
adaptive algorithms can solve the problem of finding 
\emph{any} spanning forest of a simple graph faster than 
the problem of finding a minimum-weight spanning
tree of a graph with weights assigned to its edges.

\item[\bf IV.] 
{\sf Graphs with weights and randomized algorithms.}
We give a randomized algorithm that finds a minimum-weight spanning 
forest~$T$ in expected time $\cO(\Min{m}{|T|+W\log m})$, 
where $W$ is the number of different weights on the edges in~$T$.
This shows that, for the problem of finding minimum-weight spanning forests,
randomized algorithms are provably more efficient than deterministic 
ones when the number~$m$ of edges is $\omega(n\log n)$. 

\item[\bf V.] 
{\sf Graphs specified dynamically by adversaries.}
We develop a deterministic algorithm that finds a spanning forest
of a simple graph in time $\cO(|T|\log m)$, 
if an adversary can control the time when stations holding 
edges are activated and join the computation.
We also prove that for any deterministic adaptive
algorithm~$\cA$ and an input graph containing a forest
with~$m$ edges there is a strategy of an adversary that
forces~$\cA$ to perform $\Omega(m\log m)$ steps on~$G$. 
This shows that the time performance $\Theta(m\log m)$, in
terms of~$m$, is optimal among deterministic adaptive 
algorithms in the dynamic adversarial model.
\end{enumerate}

\Paragraph{Related work.}

Early work on multiple-access channels has concentrated on
distributed protocols to handle bursty traffic of packets 
carrying dynamically generated messages. 
It included the development of protocols like Aloha~\cite{Abramson-TIT85} 
and Ethernet~\cite{MetcalfeB76}.
A survey by Gallager~\cite{Gallager85} covers the results up to 1985,
recent papers include those by 
Goldberg et al.~\cite{GoldbergMPS-JACM00}, H\aa stad et al.~\cite{HastadLR-SICOMP96},
and by Raghavan and Upfal~\cite{RaghavanU-SICOMP98}.
A multiple-access channel can be viewed as a special case 
of a multi-hop radio network. 
It is actually single-hop, since one step is
sufficient to have a message delivered between any two stations.
For an overview of work on communication in radio networks,
including the single-hop ones, see a survey by Chlebus~\cite{Chlebus-chapter-2001}.

In static problems on multiple-access channels, we assume 
that the input has been given to the stations at the
start of a protocol.
One of the most natural such problems is that of {\em selection,}
where some among the $N$~stations are given messages 
and the goal is to have any of them heard on the channel.
Willard~\cite{Willard-SICOMP86} developed a randomized protocol for this
problem that operates in expected time~$\cO(\log\log N)$, 
for a channel with collision detection.
Solving this problem requires expected time $\Omega(\log N)$, 
if detection of collision is not available, as was shown by Kushilevitz and Mansour~\cite{KushilevitzM-SICOMP98}, hence there is an exponential gap between
these two models.
In the {\em all-broadcast\/} problem we are asked to have all the
messages, given to some stations selected among the $N$~ones,
heard on the channel.
If collision detection is available then
this can be done deterministically with a logarithmic overhead
per message  by an algorithm of Koml\'{o}s and  Greenberg~\cite{KomlosG-TIT85}, 
which is optimal as shown by Greenberg and Winograd~\cite{GreenbergW-JACM85}.

G{\k a}sieniec et al.~\cite{GasieniecPP-JDM01} studied the wakeup
problem, which is a dynamic version of the selection problem.
The goal again is to have a successful transmission as soon as
possible, but the timing of stations joining the
protocol is controlled by an adversary.
Paper~\cite{GasieniecPP-JDM01} shows how efficiency depends on various levels of
synchrony, and it compares randomized and deterministic
solutions.
Other related work on radio networks has been done by Jurdzi\'nski 
et al.~\cite{JurdzinskiKZ-PODC02} and by Jurdzi\'nski and Stachowiak~\cite{JurdzinskiS02}.

All the problems mentioned above concern the communication itself
on a multiple-access channel.
Some research has also been done concerning distributed
algorithmics for specific combinatorial or optimization
problems when the underlying communication is
implemented by a multiple-access channel.
Martel and Vayda~\cite{Martel-IPL94,MartelV88} studied the problem of finding the
maximum value among those stored at a subset of stations.
Chlebus et al.~\cite{ChlebusKL-DC06} and
Clementi et al.~\cite{ClementiMS02} considered the problem of
performing a set of independent unit tasks, when the stations may
fail by crashing.
This problem is called Do-All, it was first studied in a
message-passing model by Dwork et al.~\cite{DworkHW98}.

Distributed algorithms finding minimum-weight spanning trees
have been studied before in other models.
The most popular among them assumes that the processors
are vertices and that the communication links are edges 
of an input graph. 
Such a setting  provides a unique combination in which the underlying 
communication network is also an input.
The problem of finding a minimum-weight spanning tree in such a
model was first proposed by Gallager et al.~\cite{GallagerHS-TPLS83}.
Awerbuch~\cite{Awerbuch87-STOC} developed an algorithm working in time
$\cO(n)$, where $n$ is the number of vertices.
Garay et al.~\cite{GarayKP-SICOMP98} found a solution  with
performance proportional to the diameter, if the diameter is
sufficiently large. 
Faloutsos and Molle~\cite{FaloutsosM95} studied trade offs between the time
and the number of messages.
Lower bounds on the time have been given by Peleg and
Rubinovich~\cite{PelegR00} and by Lotker et al.~\cite{LotkerPP06}.
Other graph problems studied for such a distributed setting include 
finding maximal matchings, a deterministic algorithm has been 
developed by Ha\'n\'ckowiak et al.~\cite{HanckowiakKP99},
and edge coloring, a randomized algorithm has been given by
Grable and Panconesi~\cite{GrableP97}.
An algorithm finding $k$-dominating sets was given by Kutten and
Peleg~\cite{KuttenP98}.
The issues of locality in distributed graph algorithms have been
studied by Linial~\cite{Linial92}, see the book by Peleg~\cite{Peleg2000-book}
for a comprehensive coverage.

\section{Technical Preliminaries}

We consider distributed algorithms performed by stations
that communicate over a broadcast network.
See the books by Bertsekas and Gallager~\cite{BertsekasG92} and 
Tanenbaum and Wetherall~\cite{TanenbaumW10} for systematic overviews of communication
networks.
The computations are synchronous, all the stations have access to
a global clock.

\Paragraph{Adversaries.}
The stations are categorized at each step 
as either {\em active\/} or {\em passive,} and
only active stations participate in the computation.
A passive station may be {\em activated\/} at  any step, 
then it changes its status and becomes active.
Such decisions, concerning which stations to activate and when 
to do this, are made by an {\em adversary}.
In a {\em static\/} scenario the active stations are stipulated at the
start of an execution, and are not changed later by the
adversary.
In a static setting we do not mention the adversary at all.
A {\em dynamic\/} scenario involves an adversary who can decide
on the timing when each passive station is activated 
in the course of an execution.

\Paragraph{Multiple-access channel.}

The broadcast operation is implemented on a {\em multiple-access
channel.}
All the stations receive the same information from
the channel at each step, unless they are passive, this
information is said to be {\em heard\/} on the channel.
The basic property of the channel is that if only one
station attempts the broadcast operation at a step, then its message is
delivered to all the active stations by the end of the step.
We assume that the size of a message that can be heard in
one step is as large as required by the algorithm. 
In particular, in Section~\ref{adversarial-environment} the
algorithm relies on the property that a single message can
carry a list of an arbitrary subset of the set of all the
stations.

Multiple-access channels come in two variants: either {\em
with\/} or {\em without collision detection.}
In the former case, if more than one stations broadcast at a step
then all the stations can hear the {\em collision noise\/} on the
channel.
This signal is distinct from the {\em background noise,}
which is heard when no station performs a broadcast at a step.
Both kinds of noise signals are received as indistinguishable
if the channel is without collision detection. 
A multiple-access channel is a single-hop radio network, 
in the terminology of radio networks~(see~\cite{Chlebus-chapter-2001}).

In this paper we work with the collision-detection variant.
This simplifies algorithms, and we do not lose much generality, 
at least in the static case. 
Collision detection can be implemented without affecting 
the asymptotic performance, provided that the number~$m$ of edges is 
large enough.
To this end, it is sufficient first to elect a leader among the active
stations. 
Then the consecutive steps of an execution are used 
depending on their parity:
the algorithm uses the even steps $2i$, and the stations
scheduled by it to broadcast at step $2i$ repeat this at step
$2i+1$, but the leader always broadcasts a dummy message at odd steps.
Hence if no signal is received at two consecutive steps $2i$ and
$2i+1$, then this means that there was a collision at step $2i$,
since the leader's attempt to broadcast failed.
Otherwise there is no collision, because 
the dummy message of the leader is heard at step~$2i+1$,
while nothing was heard at step $2i$. 
Selecting a leader, in the static case, can be achieved 
in time $\cO(\log n)$ by a deterministic tree-like algorithm, 
see~\cite{Chlebus-chapter-2001} for a survey of solutions of related problems in
radio networks, and \cite{JurdzinskiKZ-PODC02} for a recent work.
In the dynamic case the problem of selecting a leader is
essentially equivalent to the wakeup problem, and requires time
$\Omega(m)$, as shown by G{\k a}sieniec et al.~\cite{GasieniecPP-JDM01}.

\Paragraph{Graphs and algorithms.}

We consider graph problems, the inputs are simple graphs,
possibly with weights assigned to the edges.
The goal is to find a maximal spanning forest 
of the input graph, of the smallest weight in the case of
weighted graphs.
A spanning forest~$T$ is maximal in graph~$G$ 
if adding a new edge from~$G$ to~$T$ creates a cycle.
Throughout the rest of this paper, when we refer to a spanning forest, 
we always mean a maximal one.
The size~$|T|$ of a spanning forest~$T$ is defined to be the
number of edges in it.

Let~$n$ be the number of vertices of the input graph.
This number~$n$ is assumed to be known by all the stations.
A vertex of an input graph is identified by a number in the 
interval $[1..n]$, and an edge is identified as a pair of such
numbers. 
Each station holds a single edge.
Additionally, each station is assigned a unique {\em identifier,}
referred to as its {\em ID,} which is a positive integer.
We assume that the IDs of stations form a contiguous segment
of integers $[1..m]$, but the number $m$ is not assumed to
be known by the stations at the start of an execution.
The assumption of contiguity is relied upon when we consider 
the case of sparse graphs in a static scenario,
it makes an exhaustive enumeration efficient.

We require that all the edges of the forest sought are 
{\em revealed\/} by being heard on the channel.
After the edges of a spanning subgraph have been revealed,
and this subgraph has as many connected components as the input graph, 
then its spanning forest can be determined by some direct rule.
For instance, if there are no weights on edges, 
the first spanning forest in a lexicographic ordering of sets of edges
may be designated as the output, and if there are weights on edges, 
only spanning forests of the smallest weight are considered.
We abstract from such specific rules when presenting algorithms,
since our main concern is the communication involved in
revealing a sufficiently large subgraph with as few
transmission attempts as possible.

Typically, when a station performs a broadcast, the message
contains only its input edge.
An algorithm proceeds as a sequence of {\em queries,}
each specifies the stations that are to attempt to broadcast
at the given step.
More precisely, a query is a list including IDs of stations and/or edges. 
A station with its ID equal to~$p$ attempts to broadcast its edge on 
the channel at step~$i$ if it is {\em specified by the query $Q_i$}. 
This means that either the ID~$p$ is in $Q_i$ or the
station holds an edge that is in $Q_i$.
We often refer to a station with its ID equal to~$p$ as 
``the station~$p$,'' and to a station holding some edge~$e$ as
``the station~$e$.''

An algorithm is {\em adaptive\/}  if each query $Q_{i+1}$
depends on the feedback heard on the channel when the preceding
queries $Q_1,Q_2,\ldots,Q_i$ were executed.
An algorithm is {\em oblivious\/}  if the queries $Q_1,Q_2,\ldots,Q_i,\ldots$ are all known in advance, and depend only on the number~$n$ of vertices of the input graph.
More precisely, a query of an oblivious algorithm may be
given in terms of properties of edges of a graph on $n$~vertices, 
for instance by specifying some of the endpoints, 
or refer to specific IDs of stations.
Randomized algorithms, that we develop, are also described by queries,
but each station specified by a query additionally first
performs a random experiment, and then actually broadcasts 
only if the result of the experiment was a success.
Such experiments are Poisson trials, that is, they return either
success or failure, which are independent over the steps and the
stations, the probability of success at a step is the same
for all the stations, these probabilities may vary over the
steps.

\section{Simple Graphs}

Graphs considered in this section are simple and  
they do not have weights assigned to the edges.
The input is given at the start of an execution,
and no adversaries are involved.

\subsection{Deterministic Adaptive Algorithm}
\label{daa}

We use a routine {\sc Resolve\,}($S$) to resolve conflicts 
among a set of stations~$S$ that want to reveal one among their edges
or one among their IDs.
We discuss the case of edges in detail, the case of IDs is similar.
The procedure is based on a binary-search paradigm.
Let us fix a linear ordering $I$ of all the possible edges of a
graph with vertices in~$[1..n]$, for instance the lexicographic
one.
First all the stations (whose edges are) in~$S$ broadcast
simultaneously.
If this results in silence then $S$ is empty and {\sc 
Resolve\,}$(S)=\emptyset$ is completed.
If an edge $e$ is heard then the set~$S$ is a singleton and the conflict is
resolved, with {\sc Resolve\,}$(S)=\{e\}$.
Otherwise a collision is detected.
Let $I_1$ and $I_2$ be a partition of $I$ into left and right
subintervals, respectively, determined by the median of~$I$.
In the next step all the stations in $S\cap I_1$ broadcast.
If this results in a single edge~$e$ heard on the channel then {\sc
Resolve\,}$(S)=\{e\}$.
If there is a collision then $I_1\cap S$ is searched recursively,
with set~$I_1$ replacing $I$.
If there is silence then $I_2\cap S$ is searched recursively,
with $I_2$ replacing $I$.
It takes $\cO(\log m)$ steps to resolve a conflict and hear the
smallest edge held by a station in~$S$.

%Algorithm {\sc DetSimple}. 

\begin{figure}[t]

\rule{\textwidth}{0.5pt}

\begin{center}
\begin{minipage}{\pagewidth}
\addtolength{\baselineskip}{3pt}

{\sc Input:} the number $n$ of graph vertices ; edge $e_p$ ;

{\sc Initialization:}  
	{\sc Revealed} $:=\emptyset$, 
	{\sc Cycle} $:=\emptyset$, 
	$e_p$ is waiting ;

\F

{\sf repeat}

\T {\sf if} $e_p$ is waiting {\sf then} broadcast a dummy message ; 

\T {\sf if} silence was heard in the previous step 
{\sf then} terminate ; 

\T $\{e\} = $ {\sc Resolve}\,(\,{\sc Waiting}\,) ;

\T move edge $e$ from {\sc Waiting} to {\sc Revealed} ;

\T {\sf if} there is a cycle in the graph induced by $e_p$ 
and the edges in {\sc Revealed} 

\TT {\sf then} move $e_p$ from {\sc Waiting} to {\sc Cycle} 

{\sf until} termination ;

{\sc Output:} all the revealed edges.

\end{minipage}
\end{center}

\rule{\textwidth}{0.5pt}

\centerline{\parbox{\captionwidth}{
\caption{
\label{fig-simple-deter}
Algorithm {\sc DetSimple}. 
Code for the station $p$ storing the edge~$e_p$.
} } }
\end{figure}

%%%%%%%%%%%%%%%%%%%%%%%%%%%%%%%%%%%%%%%%%%%%%%%%%%%%%%%%%%%%%%%%%%

\Paragraph{Basic algorithm.}
We start with the algorithm {\sc DetSimple}. 
It operates by iterating phases in a loop.
At all times the input edges are partitioned into three subsets: 
\begin{quote}
\B
\begin{description}
\item[{\sc Revealed:}]
the edges already revealed on the channel, each of them is
called {\em revealed} ; 
\item[{\sc Cycle:}] 
the edges that would make a cycle if added to those in
{\sc Revealed}, each of them is called {\em cycle} ;
\item[{\sc Waiting:}] 
the remaining edges, called {\em waiting} .
\end{description}
\B
\end{quote}
Each of the waiting edges could be added to {\sc Revealed} and still 
the property that {\sc Revealed} is a forest would be maintained.
A station holding a revealed, cycle or waiting edge is 
called a revealed, cycle or waiting station, respectively.
Initially the sets {\sc Revealed} and {\sc Cycle}  are empty, 
and the set {\sc Waiting} consists of all the input edges.
During one iteration, the procedure {\sc Resolve\,}({\sc Waiting}) 
is called, and the edge eventually heard on the channel is added to 
{\sc Revealed}.
Each of the remaining waiting stations checks to see 
if it is now in {\sc Cycle}, and if this is the case,
then it will never attempt to perform a broadcast in this
execution.
A pseudocode of the algorithm is given in Figure~\ref{fig-simple-deter}.

\Paragraph{Correctness.}
The correctness of algorithm {\sc DetSimple} is guaranteed by the following
invariant maintained in each iteration: 
the edges in {\sc Revealed} make a forest on the set of vertices~$[1..n]$.

\Paragraph{Performance.}
The overall cost of algorithm {\sc DetSimple} to find a
spanning forest $T$ is $\cO(|T|\log m)$ since 
{\sc Resolve} is called  $|T|$ times.

\Paragraph{General algorithm.}
We give an algorithm that works in time $\cO(m)$ 
if $m=o(n\log n)$,  it operates as follows.
The algorithm {\sc DetSimple} is run during the odd-numbered steps.
During the even-numbered steps, the stations
broadcast their edges on the channel one by one, 
in the order of their IDs. 
If there is a silence heard during an even-numbered step, 
then this is interpreted as a termination signal, since all
the stations have revealed their edges by this step.
The two processes, run during the odd-numbered and the
even-numbered steps, do not affect one another, in 
particular the edges broadcast in the even-numbered steps 
are not treated as revealed in algorithm {\sc DetSimple}.
The following theorem follows directly from the design of
this general algorithm:

\begin{theorem}
\label{thm-1}
There is a deterministic adaptive algorithm for the static model
that finds a spanning forest~$T$ of a simple graph in time 
$\cO(\Min{m}{|T|\log m})$.
\end{theorem}

\subsection{Randomized Algorithm}

%Randomized algorithm {\sc RandSimple} for simple graphs without weights.

\begin{figure}

\rule{\textwidth}{0.5pt}

\begin{center}
\begin{minipage}{\pagewidth}
\addtolength{\baselineskip}{3pt}

{\sc Input:} edge $e_p$ ;

\F

{\sc Initialization:} $a := 1$ ; edge $e_p$ is waiting ;
{\tt counter\,} $:= 0$ ;

\F

{\sf repeat}

A: \ {\sf if} $e_p$ is waiting {\sf then} broadcast it with 
		the probability $1/a$ ;

B: \ {\sf begin case } 

\TT (a) the edge $e_p$ was heard:
change the status of $e_p$ to revealed ;

\TT (b) an edge distinct from $e_p$ was heard: 
check if $e_p$ is now cycle, \\
\TTT\TTT\TTT and if so then change its status to cycle ;

\TT (c) collision was heard:  set $a := 3a$ ;

\TT (d) silence was heard: set $a := \Min{a/3}{1}$ ;

\T {\sf end case} ;

\T {\tt counter\,} $:=$ {\tt counter\,} $ + 1$ ;

C: \ 
{\sf if} $p=$ {\tt counter\,}  {\sf then} broadcast $e_p$ ;

\T {\sf begin case}

\TT (a) the edge $e_p$ was heard:
if $e_p$ is waiting 

\TTT\TTT then change its status to revealed ;

\TT (b) an edge distinct from $e_p$ was heard:
check if $e_p$ became cycle, 

\TTT\TTT and if so then change its status to cycle ;

\T {\sf end case} ;

D: \  {\sf if} $e_p$ is waiting {\sf then} broadcast a dummy message ;

{\sf until} silence was heard in step D ;

\F
{\sc Output:} the set of the revealed edges.

\end{minipage}
\end{center}

\rule{\textwidth}{0.5pt}

\centerline{\parbox{\captionwidth}{\caption{\label{fig-simple-rand}
Randomized algorithm {\sc RandSimple} for simple graphs without weights.
Code for the station~$p$ that stores the input edge~$e_p$. 
}}}
\end{figure}

We present a randomized algorithm {\sc RandSimple} that finds
a spanning forest in a graph in expected time proportional to its size.
The algorithm is similar in structure to the deterministic one, 
and we use the same terminology.
The main difference and advantage is that the expected time 
to reveal a waiting edge is constant.

A pseudo code of the algorithm is in Figure~\ref{fig-simple-rand}.
The algorithm uses a variable denoted~$a$, which is maintained 
by all the stations, to approximate the number of waiting stations 
during an execution. 
The stations broadcast with the probability~$1/a$.
They update the estimate $a$ if silence or noise is heard. 
If an edge is heard, then all the sets {\sc Revealed}, {\sc
Waiting} and {\sc Cycle} are updated as in the algorithm 
{\sc DetSimple}.

\Paragraph{Correctness.}
The correctness of algorithm {\sc RandSimple} follows from the following 
two facts. 
First, the output is a superset of a spanning forest.
This follows from the observation that as long as a~spanning forest
has not been found, then there are still some waiting edges.
Secondly, the algorithm terminates in the worst-case time
$\cO(m)$.
This is because in each iteration a specific edge is
verified directly if it is waiting, exhausting all of them in a
systematic way, except possibly for the last iteration.
When the counter variable attains the value equal to the maximum ID 
of a station and its edge is broadcast in step~C,  
then it was the last waiting edge, hence silence is heard in step~D
and the algorithm terminates.

\Paragraph{Performance.}
We will show that expected time of the algorithm 
is proportional to the number of edges in the obtained forest.

Let a {\em round\/} denote one iteration of the main loop of
the algorithm.
An execution of algorithm {\sc RandSimple} can be partitioned 
into phases, each comprising a sequence of consecutive full rounds, 
and finished by an iteration in which (at least one) edge 
in a forest has been broadcast.
An additional closing round, in which the algorithm terminates,
is also possible.
For the purpose of the probabilistic analysis 
we assume conservatively that no edges are ever broadcast 
successfully during step~C.

Let $\cW_i$ denote the set of waiting edges in the beginning
of the $i$th phase. 
In particular, the set~$\cW_1$ consists of all the~$m$ input edges. 
Let $a_i$ denote the value of variable~$a$ at the start 
of the $i$th phase; 
this number is interpreted as an approximation of the size~$|\cW_i|$ 
of the set~$\cW_i$.
The value stored in the variable~$a$ at the start of the $j$th~round 
of phase~$i$ is denoted as~$a_{i,j}$.
We say that the algorithm {\em has a good approximation\/} 
if the inequalities 
\[
\frac{1}{3}\,|\cW_i| \leq a_{i,j} \leq 3\, |\cW_i|
\]
hold.
We first estimate the probabilities of silence, noise and 
a successful broadcast, respectively, heard on the channel just after step~A 
in a single round of a phase.

\begin{lemma} 
\label{l:silence}
There is a constant $c_1>\frac{1}{2}$ such that if $a>3 |\cW_i|$,
then the probability of silence in a round of phase~$i$ is at least~$c_1$.
\end{lemma}

\begin{proof}
If $\cW_i\neq \emptyset$, then $a>3$.
Interpreting actions of stations as Bernoulli trials,
we obtain that the probability of silence is at least
\begin{eqnarray*}
\Bigl(1-\frac{1}{a}\Bigr)^{|\cW_i|} 
&\ge& 
\Bigl(1-\frac{1}{a}\Bigr)^{a/3} \\
&\ge& 
\Bigl(1-\frac{1}{4}\Bigr)^{4/3} \ ,
\end{eqnarray*}
which is larger than~$\frac{1}{2}$.
\end{proof}

\begin{lemma} 
\label{l:noise}
There is a constant $c_2>\frac{1}{2}$ such that if $a<|\cW_i|/3$, 
then the probability of noise  during a round in phase~$i$ is at
least~$c_2$.
\end{lemma}

\begin{proof}
We start by estimating the probabilities of silence and a
successful broadcast.
The probability of silence is at most 
\[
\Bigl( 1-\frac{1}{a}\Bigr)^{|\cW_i|} \le
\Bigl( 1-\frac{1}{a}\Bigr)^{3a} \le e^{-3}\ .
\]
The probability of a successful broadcast is  that of an
exactly one success in a sequence of $|\cW_i|$ Bernoulli
trials, and is at most
\begin{eqnarray*}
|\cW_i|\,\frac{1}{a}\,\Bigl( 1-\frac{1}{a}\Bigr)^{|\cW_i|-1}
&=&
\frac{|\cW_i|}{a}\,\frac{a}{a-1}\,\Bigl( 1-\frac{1}{a}\Bigr)^{|\cW_i|}\\
&\le &
2\,\frac{|\cW_i|}{a}\,\Bigl(1-\frac{1}{a}\Bigr)^{a\,|\cW_i|/a}\\
&\le &
6\,e^{-3}\ .
\end{eqnarray*}
Since the probability of noise is~1 minus the probabilities
of silence and of a successful broadcast, we obtain the
following estimate as a lower bound:
\[
1-(e^{-3}+6e^{-3})=1-7\,e^{-3}\ ,
\]
and it is larger than~$\frac{1}{2}$.
\end{proof}

\begin{lemma} \label{l:broadcast}
There is a constant $c_3>0$ such that if both the inequalities 
$|W_i|/3\le a\le 3|\cW_i|$ hold then the probability of a
successful broadcast  during a round in phase~$i$ is at
least~$c_3$.
\end{lemma}

\begin{proof}
We estimate the probability of exactly one success
in a sequence of $|\cW_i|$ Bernoulli trials:
\begin{eqnarray*}
\frac{|\cW_i|}{a}\,\Bigl( 1-\frac{1}{a}\Bigr)^{|\cW_i|-1}
&=&
\frac{|\cW_i|}{a}\,\frac{a}{a-1}\,\Bigl( 1-\frac{1}{a}\Bigr)^{|\cW_i|}\\
&\ge&
\frac{1}{3}\,\Bigl( 1-\frac{1}{a}\Bigr)^{3a}\ ,
\end{eqnarray*}
which is at least $1/(3e^3)$.
\end{proof}

We model an execution of the algorithm as a combination of
two random processes.
One is a discrete-time random walk with a retaining barrier:
the nonnegative integers are the possible coordinates of a
particle, the barrier is at the origin with the coordinate equal
to zero.
Each time the particle is at the origin then the second
process is started which is just a single Bernoulli trial.
The random walk terminates after the first success and is
restarted at some positive integer coordinate.
Details are as follows.

If the algorithm has a good approximation, then this is
interpreted as the particle being at the barrier.
If either the inequalities
\[
\frac{|\cW_i|}{3^{k+1}} \le a_{i,j} < \frac{|\cW_i|}{3^k}
\]
or the inequalities
\[
3^{k} \, |\cW_i| < a_{i,j} \le  3^{k+1} \, |\cW_i| 
\]
hold, for $k\ge 1$, then this is interpreted as the particle
being at the distance $k$ from the barrier.
If the particle is at the origin then a Bernoulli trial is
performed with the probability of success equal to the
number~$c_3$ of Lemma~\ref{l:broadcast}.
A success is interpreted as a successful broadcast, which
starts a new phase~$i$.
After a success, the particle is moved to the location
determined by the current $a_i$ and $\cW_i$, and
after a failure, the particle is moved to location~1.
If the particle is at the location with a coordinate~$k>0$,
then it moves either to $k-1$ or to $k+1$ in the next step,
according to the following rules. 
It moves to $k-1$ with the probability
$c_0=\Min{c_1}{c_2}>\frac{1}{2}$, where $c_1$ and $c_2$ are as in
Lemmas~\ref{l:silence} and~\ref{l:noise}.
It moves to $k+1$ with the probability~$1-c_0<\frac{1}{2}$.
All the random moves are independent of each other.

This underlying random-walk model captures the behavior of
the algorithm for the following two reasons.
First, the moves of the particle correspond to the modifications of
the variable~$a$ in steps~B(c) and~B(d) in Figure~\ref{fig-simple-rand}.
Second, any positive integer is always between two consecutive
powers of~3; take the number~$k$ such that $3^k\le |\cW_i|<3^{k+1}$, for
integer~$k\ge 0$, and if $a=3^k$ or $a=3^{k+1}$, then this
is reflected in the model by the particle located at the barrier. 
Our probabilistic analysis is based on the following property 
of such a random walk on a discrete axis with a barrier:
if a particle starts at location $\ell>0$, then the expected
time needed to reach the barrier is $\cO(\ell)$, where the
constant hidden in the notation~$\cO$ depends on the
probability~$c_0$, see~\cite{GrimmettS2001-book}.

\begin{theorem}
\label{t-randsimple}
Algorithm {\sc RandSimple} works in expected time proportional
to the size~$|T|$ of a spanning forest~$T$ it finds, 
and in the worst-case time $\cO(m)$.
\end{theorem}

\begin{proof}
Let $L_i$ be the distance from the origin in the beginning
of phase~$i$.
Then $\sum L_i = \cO(\log m)$ because the sizes of the set 
of waiting edges are monotonically decreasing.
It follows that the expected total time spent by the particle between
the beginnings of phases and then reaching the origin for
the first time is~$\cO(\log m)$.
This amount of time is not asymptotically more than the size
of the spanning forest, which is~$\Omega(\sqrt{m})$.

Let $X$ be a random variable equal to the number of steps it
takes the particle to reach the barrier at the origin
after a start at position~1.
Let $Y$ be the number of attempts in a sequence of Bernoulli
trials, each with the probability~$c_3$ of success, before 
a success occurs.
Then the expected length of a phase is at most
\[
\sum_{k\ge 1} \,(k\cdot \mE [X]) \cdot \Pr (\,Y=k\,)
\, =\,
\mE [X]\cdot \sum_{k\ge 1} \, k \cdot \Pr (\,Y=k\,)
\, =\,
\frac{1}{c_3} \cdot \mE [X] \ ,
\]
which is $\cO(1)$.
This completes the proof for the expected performance,
since a phase contributes an edge in a spanning forest. 
The worst-case upper bound is straightforward.
\end{proof}

A spanning forest of a graph with $n$ nodes an $m$ edges 
has at most $\Min{m}{n-1}$ edges, which gives another
possible form of a performance bound following directly from
Theorem~\ref{t-randsimple}.

%%%%%%%%%%%%%%%%%%%%%%%%%%%%%%%%%%%%%%%%%%%%%%%%%%%%%%%%%%%%%%%%%%%%

\subsection{Lower Bound for Simple Graphs}
\label{lb-oblivious}

In this subsection we prove a lower bound~$\Omega(m)$ on the number
of queries required by a deterministic oblivious algorithm, when
the input graph has no weights assigned to its edges.
This shows that there is a gap in the optimum performances
between adaptive and oblivious algorithms, among deterministic
ones.

Suppose $\cA$ is an algorithm and $m$ is the number of edges we
would like to be the size of an input.
To simplify exposition, we assume that the queries in~$\cA$
contain only edges and no IDs of stations. 
Our goal is to construct a connected graph $G=(V,E)$, with the set of nodes
$V=\{1,\ldots,n\}$, for a suitable number~$n$, 
and the set~$E$ of $m$~edges, such that 
there is a vertex $v\in V$ with the property that, for each vertex
$w\in V$ different from $v$, and any query $Q_i$, for $1\le i\le k$, 
we have $Q_i\cap E\ne \{(v,w)\}$. 
Such a graph would make the algorithm~$\cA$ perform more than $k$~steps,
because after these many steps no edge with~$v$ as its endpoint 
would have been heard. 

The following specification of an input graph for~$\cA$ 
is referred to as the {\em five-phase construction}.
Take a number~$n$ such that the inequalities 
$2n\le m\le n(n-1)/4$ hold.
Consider the first $k$~queries $Q_1,\ldots,Q_k$ of algorithm~$\cA$,
where $k=\lfloor\frac{m-n}{2}\rfloor$. 
Notice that $k=\Omega(m)$.
We start from a configuration such that, 
for each edge in the complete graph on $n$~vertices, 
there is a station holding this edge. 
We specify which of these stations, and hence also edges, 
should be removed, so that the number of the remaining edges is exactly~$m$.
Formally, we proceed through a sequence of five phases 
and maintain three sets of edges: $Q$, $F$ and~$T$.
The set~$Q$ is initialized to $Q:=\bigcup_i Q_i$,  
it contains all the edges that can appear in any of the queries. 
The set~$F$ is initialized to all the possible edges between the
$n$~vertices.
The set~$E$ of edges of the input graph~$G$, that we will identify 
in the course of the construction, will be stored in the set~$T$.

\Paragraph{\sf Phase 1} (preventing edges to be heard 
for as long as possible) 

\begin{quote}
\BB

{\tt stop} $:=0$ ; 

{\sf while} {\tt stop} $=0$ {\sf do} {\sf begin}

\T {\tt stop} $:= 1$ ; 

\T {\sf if} there is a $Q_i$ such that 

\TT \TT ( $|Q_i\cap F|=1$ ) {\sf and } 
( there is no isolated vertex in $F \setminus Q_i$ )

\T {\sf then} {\sf begin} $F:=F\setminus Q_i$ ; {\tt stop} $:=0$
{\sf end} ; 

{\sf end}
\B
\end{quote}

\noindent
{\em Comment:}
We remove from $F$ those edges that do not isolate any vertex
but would be heard if the algorithm asked about them and 
they were in the graph.
The purpose is to delay the step when the algorithm learns about 
a pair of vertices being connected. 
If it keeps asking about specific neighbors of a vertex, then
it hears nothing except for the last possible moment.
We keep passing through all the queries until we are sure that
the last modification has not created new singletons $Q_i\cap F$ 
that would be heard.
It follows that the graph spanned by the set of edges~$F$ is connected.

\Paragraph{\sf Phase 2} (initialization) 

\begin{quote}
\B
Set $T$ equal to a spanning tree of the graph spanned by the
set~$F$ of edges.
\end{quote}

\noindent
{\em Comment:}
This is just an initialization of the set $T$. 
We will keep increasing it during the next phases 
as long as its size is not large enough.

\Paragraph{\sf Phase 3} (adding collisions) 

\begin{quote}
\BB
{\sf for} $i:=1$ {\sf to} $k$ {\sf do}

\T {\sf if} $Q_i\cap T=\{e_1\}$ {\sf and} $|Q_i\cap F|>1$ {\sf then} 

\T \T{\sf begin}
choose $e_2\in Q_i\cap F$ such that $e_2\ne e_1$ ; 
$T:=T\cup \{e_2\}$
{\sf end} ;

\T {\sf if} $|Q_i\cap T|=0$ and $|Q_i\cap F|>1$ {\sf then} 

\T \T{\sf begin}
choose $e_1,e_2\in Q_i\cap F$ such that $e_1\ne e_2$ ; 
$T:=T\cup \{e_1,e_2\}$
{\sf end} ;

{\sf endfor}
\B
\end{quote}

\noindent
{\em Comment:}
Now the set $T$ makes the algorithm hear collisions at all the steps
when the set~$F$ does.

\Paragraph{\sf Phase 4} (increasing size) 

\begin{quote}
\BB
{\sf for} $i:=1$ {\sf to} $k$ {\sf do}

\T {\sf if} $|T|<m$ and $Q_i\cap T\ne \emptyset$ {\sf then}
{\sf begin }

\TT choose $X\subseteq Q_i\cap F$ of a maximal size 
such that $|T\cup X| \le m$ ; $T:=T\cup X$ 

\T {\sf end }
\B
\end{quote}

\noindent
{\em Comment:}
In this phase we increase the size of~$T$, by adding edges 
included in queries, to be as close as possible to~$m$.
This phase is productive if the set~$Q$ is large enough,
otherwise we still have Phase~5 as the last resort.

\Paragraph{\sf Phase 5} (padding with edges not mentioned)

\begin{quote}
\BB
{\sf If} $|T|<m$  {\sf then} add $m-|T|$ edges to $T$ that are not in $Q$.
\B
\end{quote}

\noindent
{\em Comment:}
This phase is needed in case algorithm~$\cA$ has the queries involving
only a small set of edges.

This completes the five-phase construction of the graph~$G$. 
We say that a vertex~$v$ is a {\em $Y$-witness,} for a set~$Y$ of edges,
if none of the edges~$(v,w)$ satisfies $Q_i\cap Y=\{(v,w)\}$, 
for any number~$i\le k$  and for any vertex~$w$ in~$V$ distinct from~$v$. 
Correctness of the five-phase construction 
is formulated as Lemma~\ref{l-witness}.

\begin{lemma}
\label{l-witness}
The five-phase construction results in obtaining a set~$T$ of
size~$m$ and with a vertex that is a $T$-witness.
\end{lemma}

\begin{proof}
We examine the five phases one by one.

\Paragraph{\sf After Phase 1.}

Suppose to the contrary that there are no $F$-witness vertices.
We have $Q_i\cap F=\{e\}$ holds only if $F\setminus\{e\}$ has some 
isolated vertex. 
An edge~$(v,w)$, such that $Q_i\cap F=\{(v,w)\}$, is not
removed only if otherwise one among the vertices $v$ and $w$
would become isolated in~$F$.
Notice also that at least
\[
(n-1)+(n-2)+\cdots+(n-l)=l\cdot (2n-l-1)/2
\]
edges need to be removed in order to isolate $l$~or more vertices,
which follows by induction on~$l$.
Removing one edge results in the intersection of at least one query 
with~$F$ being empty.
We can remove at most $k$~edges. 
Should we like to restrain all the $n$~vertices from being  $F$-witnesses
we would need to isolate at least $n/2$~vertices, thus
needing to remove at least
$n/2 (2n-n/2-1)/2=n(\frac{3}{2}n-1)/4=\frac{3}{8}n^2-n/4$ edges from~$F$. 
Because of the inequality $k<n^2/8-n/2$, 
at least one vertex~$v$ has not been prevented from being an $F$-witness.

It follows that all the singleton intersections that contain 
an edge adjacent to vertex~$v$ have been removed from~$F$.
The inequality $|F|\ge |Q|-k$ follows by the fact that an edge is
removed from the set~$F$ only if at least one query~$Q_i$ becomes
disjoint with the updated~$F$.

\Paragraph{\sf After Phase 2.}

The tree~$T$ contains each edge~$(v,w)$
such that $Q_i\cap F=\{(v,w)\}$, for some~$i\le k$. 
This follows from the observation that such an edge $(v,w)$ 
has not been removed in Phase~1, hence one among the vertices~$v$ and~$w$ 
is of degree~1 in $F$, and the edge~$(v,w)$ was put into
set~$T$.
Clearly, $|T|\le n-1$, and hence there is a $T$-witness at
this point.

\Paragraph{\sf After Phase 3.}

Set $T$ has at most $n-1+2k<m$ edges. 
This is because  at most two edges can be added to the
set $T$ in each step $i\le k$.
The vertex~$v$ that is a $T$-witness
is the same as in the analysis of the previous phases.
More precisely, there was no edge~$(v,w)$ such that $Q_i\cap T=\{(v,w)\}$
in the beginning of Phase~3, the set~$T$ was not decreased during Phase~3, 
and no edge such that $Q_i\cap F=\emptyset$ was added in step~$i$. 

\Paragraph{\sf After Phase 4.}

We obtain set $T$ of a size that is between $|Q|-k$ and~$m$. 
This is because we can add to $T$ any number 
of edges from $F\setminus T$, as long as the inequality $|F|\ge |Q|-k$
holds.
There is a $T$-witness because during Phase~4 we add only 
edges in such queries that have at least one element in~$T$.

\Paragraph{\sf After Phase 5.}

The property of an existence of a $T$-witness is
maintained because Phase~5 does not interfere 
with intersections of $F$ and $Q_i$.
Since the size~$|T|$ was between $|Q|-k$ and~$m$, we can add
sufficiently many elements from the set 
$\{(v,w):v,w\in V, v\ne w\}\setminus Q$ 
of size $n(n-1)/2-|Q|\ge m+k-|Q|$ to obtain $|T|=m$.
\end{proof}

Lemma~\ref{l-witness} contains all the essential ingredients 
of a lower bound in terms of the number~$m$ of edges alone.
We show a more general fact formulated in terms of both the numbers
of nodes and edges.

\begin{theorem}
\label{thm-lower-bound-oblivious}
For any deterministic oblivious algorithm~$\cA$
finding spanning forests, and numbers~$n$ and~$m$, where~$m=\cO(n^2)$, 
there exists a connected graph $G_\cA$ 
with $n$~vertices and $\Theta(m)$~edges,
such that algorithm~$\cA$ requires time $\Omega(m)$ 
to find a spanning tree of~$G_\cA$.
\end{theorem}
\begin{proof}
If the inequality $2n\le m\le n(n-1)/4$ holds, then apply 
the five-phase construction. 
It follows from Lemma~\ref{l-witness} that if we take 
the final~$T$, obtained by completing the five-phase construction,
as the set of edges~$E$, then the oblivious algorithm~$\cA$ needs 
more than $k=\Omega(m)$ steps to broadcast a spanning tree 
of the graph~$G=(V,E)$. 
This is because a $T$-witness has not been 
heard by the $k$th step as an endpoint of an edge.

In the remaining cases of dependencies between the
numbers~$m$ and~$n$, we can proceed as follows.
If $n(n-1)/4 < m$, then we can use the same construction
as if the number~$m$ were equal to~$n(n-1)/4$.
If $m < 2n$, then a simple path of length $\ell=\Min{n-1}{m}$ does the
job, because any algorithm needs to broadcast at least $\ell$ 
edges to reveal the whole path.
\end{proof}

There is a simple oblivious algorithm that operates in time~$m+1$:
the station~$i$ broadcasts its edge at step~$i$.
Theorem~\ref{thm-lower-bound-oblivious} may be interpreted
as follows:
listing all the edges systematically is asymptotically optimal 
among oblivious algorithms. 

\section{Weighted Graphs}

In this section we consider graphs that have positive
weights assigned to their edges.
The input is specified at the start of an execution, 
no adversaries are involved.

\subsection{Randomized Algorithm}

%Randomized algorithm {\sc RandWeighted} for graphs with weights.

\begin{figure}

\rule{\textwidth}{0.5pt}

\begin{center}
\begin{minipage}{\pagewidth}
\addtolength{\baselineskip}{3pt}

{\sc Input:} edge $e_p$ of weight $w_p$ ;

\F

{\sc Initialization:} 
edge $e_p$ is waiting ; 
\F

{\sf repeat}

\ \ \ \ \ $a := 1$ ; {\tt weight} := $+\infty$ ;

A:  {\sf repeat}

\TT {\sf if} ($e_p$ is waiting) and $(w_p<\text{\tt weight})$

\TT\TT{\sf then} broadcast $w_p$ with the probability $1/a$ ;

\TT{\sf begin case } 

\TTT (a) some weight $w$ was heard  :  set {\tt weight} $:= w$ ;

\TTT (b) collision was heard  :  set $a := 3a$ ;

\TTT (c) silence was heard   :  set $a := \Min{a/3}{1}$ ;

\TT {\sf end case} ;

\TT {\sf if} ($e_p$ is waiting) and ($w_p<\text{\tt weight}$)
{\sf then} broadcast a dummy message ; 

\T {\sf until} silence was heard in the previous step ;

B: \ {\sf repeat}

\TT {\sf if} ($e_p$ is waiting) and $(w_p = \text{\tt weight})$

\TT\TT{\sf then} broadcast $w_p$ with the probability $1/a$ ;

\TT {\sf begin case}

\TTT (a) the edge $e_p$ was heard  :  
change the status of $e_p$ to revealed ;

\TTT (b) an edge different from $e_p$ was heard : 
check if $e_p$ became cycle, 

\TTT\TTT and if so then change its status to cycle  ;

\TTT (c) collision was heard :  set $a := 3a$ ;

\TTT (d) silence was heard  :  set $a := \Min{a/3}{1}$ ;

\TT {\sf end case} ;

\TT {\sf if} ($e_p$ is waiting) and $(w_p = \text{\tt weight})$
{\sf then} broadcast a dummy message ;

\T {\sf until} silence was heard in the previous step ;

\T {\sf if} $e_p$ is waiting {\sf then} broadcast a dummy message ;

{\sf until} silence was heard in the previous step ;

\F
{\sc Output:} the set of the revealed edges.

\end{minipage}
\end{center}

\rule{\textwidth}{0.5pt}

\centerline{\parbox{\captionwidth}{
\caption{
\label{fig-weighted-rand}
Randomized algorithm {\sc RandWeighted} for simple graphs with weights.
Code for the station~$p$ that stores the input edge~$e_p$ with
weight $w_p$. 
}}}
\end{figure}

%%%%%%%%%%%%%%%%%%%%%%%%%%%%%%%%%%%%%%%%%%%%%%%%%%%%%%%%%%%%%%%%%

We present a randomized algorithm that finds a  minimum-weight 
spanning forest of a~graph. 
It operates in expected time $\cO(|T|+ W\log m)$, where
$W$ is the number of different weights assigned to the edges 
in~$T$.
The algorithm is called {\sc RandWeighted}, its
pseudocode is given in Figure~\ref{fig-weighted-rand}.

The algorithm is a generalization of {\sc RandSimple},
it uses the same categories of edges, and the same
variable~$a$ serving as a stochastic estimate of a set of
edges.
There is an additional variable {\tt weight} 
interpreted as an edge weight.
Algorithm {\sc RandWeighted} is a loop in which two inner
loops~A and~B are executed.
The purpose of the first loop~A is to find the smallest
weight of a waiting edge.
This loop terminates with the variable {\tt weight} storing
this value.
Then the next loop~B follows in which a maximal set of
edges is found such that each of these edges is of a 
weight equal to that stored in the variable {\tt weight}, and
all can be added to the current set of the revealed edges
with the property that it is a forest being preserved.

Similarly as before, an edge is waiting if it would not create 
a cycle if added to the revealed part of a minimum-weight forest.
We depart from the previous terminology in this section, in that 
an edge is said to be \emph{revealed} only if it was broadcast as 
a {\em minimum-weight\/} one among those that were waiting.
The remaining edges, broadcast during the selection in the
inner loop~A, are {\em not\/} categorized as revealed. 
This is in contrast with algorithms from the preceding sections,
where \emph{every} waiting edge broadcast successfully
became revealed immediately.

\Paragraph{Correctness.}
The algorithm {\sc RandWeighted} is essentially an implementation
on a multiple-access channel of a greedy minimum-spanning-forest algorithm
(see~\cite{CormenLRS09-3rd-edition-book}),
hence a set of the revealed edges is a minimum-weight spanning forest of 
the input graph.
The algorithm terminates with probability~$1$ because its
expected time is finite, as we show next.

\Paragraph{Performance.}
Our analysis of the behavior of algorithm {\sc RandWeighted}
is an extension of that for {\sc RandSimple}.

\begin{lemma}
\label{l-randweighted}
The randomized algorithm {\sc RandWeighted} finds a minimum-weight spanning
forest~$T$ in expected time $\cO(|T|+ W\log m)$,
where $W$~is the number of distinct weights on the edges
of~$T$.
\end{lemma}

\begin{proof}
Consider an iteration of the main loop.
The purpose of the first inner loop~A is to find the minimum weight 
among all the waiting edges.
There are two phenomena here that occur concurrently.
One is a randomized binary search that governs the selection
of weights.
The expected number of selections made is $\cO(\log m)$,
which follows from the fact that the expected height of a
random binary-search tree is logarithmic in the number of its
leaves, see for instance~\cite{CormenLRS09-3rd-edition-book}.
The other phenomenon is similar to the behavior of algorithm
{\sc RandSimple}, as modeled by a discrete random walk in
the proof of Theorem~\ref{t-randsimple}.
After each new weight has been broadcast, there is an adjustment
needed to the variable~$a$ to catch up with the decreasing
size of the set of the edges of smaller weights, if there
are still any.
This corresponds to placing the particle at a place possibly
distant from the origin in the terminology of the proof of
Theorem~\ref{t-randsimple}.
The expected total time spent on such catching up is
$\cO(\log m)$, since the subsets keep decreasing.
Except for that, the expected time spent on producing 
a new edge is $\cO(1)$.
These two $\cO(\log m)$ bounds add up, and this is the cost of
producing the first revealed edge, of the minimum weight
among the waiting ones, when the second inner loop~B is
executed for the first time in the given iteration of the
main loop.
Loop~B may need to be repeated more times as long as there are
waiting edges of the same weight, but each of them is
produced with expected time~$\cO(1)$ per edge within the same
iteration of the main loop.
\end{proof}

\begin{theorem}
\label{thm-4}
There is a randomized algorithm that finds a minimum-weight spanning
forest~$T$ in expected time $\cO(|T|+ W\log m)$,
where $W$~is the number of distinct weights on the edges in~$T$,
and in worst-case time~$\cO(m)$.
\end{theorem}

\begin{proof}
To guarantee the claimed worst-case performance,
we apply a similar stratagem as in algorithm {\sc RandSimple}
for simple graphs. 
We have all the edges broadcast on the
channel in a systematic way, until a forest of minimum
weight has been found or all the edges have been exhausted,
whatever comes first.
More precisely, let the algorithm {\sc RandWeighted} be run 
in odd-numbered steps, while in even-numbered steps 
the stations broadcast their edges and weights in order of their~IDs. 
Unlike algorithm {\sc RandSimple},
the processes should be independent of each other, in the sense
that the edges broadcast in the odd-numbered steps are not
categorized as revealed, cycle or waiting, until all have
been exhausted and the algorithm stops with all the edges
broadcasted.
This allows us to combine the expected performance of {\sc RandWeighted}
given in Lemma~\ref{l-randweighted} with a worst-case 
upper bound proportional to the number~$m$ of all the edges.
\end{proof}

\subsection{Lower Bound for Weighted Graphs}
\label{lb-weighted}

In this section we prove a lower bound $\Omega(m)$ 
for deterministic adaptive algorithms for weighted graphs.
The lower bound is of the same form as in Section~\ref{lb-oblivious},
the difference is that the algorithms are adaptive rather
than oblivious, while the graphs are weighted rather than simple.

For each deterministic adaptive algorithm~$\cA$, we construct 
a certain weighted graph $G_\cA$ of $n$~vertices and $m$~edges. 
We start with any assignment of edges to the stations at the
very beginning of computation, so that the 
graph $G=(V,E)$ has no isolated vertices. 
Let $\mathbb{N}$ be the set of positive integers.
In the construction we use only weights from set $\{1/j:j\in \mathbb{N}\}$.
Our goal is to assign weights to all the edges.
Each station~$i$ will have a set $A_i(t)$ of numbers in
$\{1/j:j\in \mathbb{N}\}$ assigned to it after step~$t$
of the construction. 
Initially we set $A_i(0)=\{1/j:j\in \mathbb{N}\}$,
for every station $i$.

Since there is a one-to-one correspondence between the edges and 
the stations, we treat $E$ also as a set of stations.
Denote by $E(t)$ the set of stations which have an edge with some
weight assigned to it by step $t$ of algorithm~$\cA$. 
Each station $i$ from $E(t)$ has just one element in~$A_i(t)$. 
We assume the following invariant after step $t$: 
for each station $i\in E\setminus E(t)$, the set $A_i(t)$ is infinite.

Consider step $t+1$ of algorithm~$\cA$, 
as determined by step~$t$ of the construction.
Let $S_0(t+1)$ denote the set of these stations $i$ that do not 
broadcast during step $t+1$ of algorithm~$\cA$, for 
an infinite number of possible weights from $A_i(t)$. 
Similarly, let $S_1(t+1)$ denote the set of stations that 
broadcast in step $t+1$, for infinitely many possible 
initial values. 

Step $t+1$ of the construction is broken 
into the following cases:

\noindent
{\sc Case 1:} $S_1(t+1)\setminus S_0(t+1)=\emptyset$. 

\noindent
Make set $A_i(t+1)$ contain all 
the weights from $A_i(t)$, for which station~$i$ does 
not broadcast during step $t+1$ of algorithm~$\cA$, 
for each station $i\in S_0(t+1)$. 
Hence $E(t+1)=E(t)$ and the invariant holds after step $t+1$.
Moreover, there are no new stations broadcasting on the channel, 
hence we can determine what happens in step $t+1$ by
considering only stations in~$E(t)$.

\noindent
{\sc Case 2:} $S_1(t+1)\setminus S_0(t+1)=\{i\}$. 

\noindent
We set $A_i(t+1)=\{1/\ell\}$, where $1/\ell$ is
the maximum weight in $A_i(t)$ for which station~$i$ broadcasts 
in step $t+1$ of algorithm~$\cA$. 
For the other stations $j$ in $E\setminus E(t)$, we set
$A_j(t+1)=A_j(t)\cap (0,1/\ell)$. 
In this case $E(t+1)=E(t)\cup \{i\}$ and the invariant holds
after step $t+1$. 
In step $t+1$ only station $i$ broadcasts,
among stations in $E\setminus E(t)$, hence by an argument similar 
to that used in Case~1, we have control over what happens in step~$t+1$. 

\noindent
{\sc Case 3:} $|S_1(t+1)\setminus S_0(t+1)|>1$. 

\noindent
We choose two different stations 
$i_1,i_2\in S_1(t+1)\setminus S_0(t+1)$ and decide on the sets to
be $A_{i_1}(t+1)=\{1/\ell_1\}$ and $A_{i_2}(t+1)=\{1/\ell_2\}$,
similarly as for vertex~$i$ in Case~2. 
For the other stations~$j$ in~$E\setminus E(t)$ we set 
$A_j(t+1)=A_j(t)\cap (0,1/\max(\ell_1,\ell_2))$. 
In step $t+1$ only stations~$i_1$ and~$i_2$ broadcast,
among all the stations in~$E\setminus E(t)$, hence, by a similar argument 
as in Case~1, we can decide on the events in step~$t+1$.

This construction gives the following result:

\begin{theorem}
For each adaptive deterministic algorithm~$\cA$ 
that finds a minimum spanning forest, and a possible number
of edges~$m$, there is a graph $G_\cA$ 
such that algorithm~$\cA$ terminates after $\Omega(m)$~steps
when given the graph~$G_\cA$ as input. 
\end{theorem}

\begin{proof}
Consider the first $m/2$ steps and graph $G_\cA$ constructed as
described above. 
Our construction has the following properties.
First, in one step of construction at most two stations (edges) can have 
assigned weights; each among such stations $i$ is moved to $E(t)$ 
and has  $A_i(t)$ of unit size. 
Second, for any station~$j$ not in $E(t)$ we have $|A_j(t)|=\infty$, 
hence in the next steps some sufficiently small weight could 
be assigned to~$j$.
It follows, by induction on the number of steps~$t$ of algorithm~$\cA$ 
on graph $G_\cA$, that if $t<m/2$ then an edge with the minimum weight 
has not been broadcast successfully on the channel. 
\end{proof}

\section{Adversarial Environment}
\label{adversarial-environment}

In this section we consider dynamic graphs
with no weights assigned to their edges.
There is an adversary who is able to decide on timing when the
station holding a particular edge is activated.
An activated station is aware of being activated at the given
step, and of the number of the step, 
as counted by a global clock.
Multiple stations may be activated at a step, or none.
We assume that the global clock is started exactly at the
first step of the algorithm, and that at least one station
is active then.

The stations are activated by an adversary but they halt on
their own.
The issue of termination and correctness needs to be
clarified precisely, since an adversary might activate a
number of stations just before the stations already active
have decided to terminate.
We say that the algorithm {\em terminates\/} at the step 
when some station successfully broadcasts a special 
{\em termination signal.}
Our approach to correctness is based on disregarding the edges held
by the stations activated too late. 
To make this precise, we call the period of the last $c$~steps by
termination the {\em closing $c$~steps,} for any fixed 
integer~$c>0$.
An algorithm is said to be {\em $c$-correct,} for a positive
integer~$c$, if the output is a spanning forest 
for the input graph having the edges that are held by the 
stations activated before the closing $c$~steps.
An algorithm is {\em correct\/} if it is $c$-correct for
some integer~$c>0$ and a sufficiently large number~$n$ of vertices.
We consider only the $m$ edges held by the stations activated 
before the last $c$~closing steps, 
where parameter~$c$ is some fixed constant,
as required in the definition of correctness.

First we prove a lower bound for deterministic adaptive algorithms in 
this model.
This lower bound is strong in the sense that it holds for
arbitrary graphs containing a forest of a given size.

\begin{theorem}
\label{t-adversarial-lower}
For any positive constant~$c$ and an adaptive
deterministic $c$-correct algorithm~$\cA$ that is able to find 
spanning forests in simple graphs, and for any positive integer~$m$
that is sufficiently large depending on~$c$, 
and for any simple forest~$G$ with $m$~edges, 
there is such a strategy of an adversary to activate stations 
that forces algorithm~$\cA$ to perform $\Omega(m\log m)$ 
steps on the input~$G$ if run against the strategy.  
\end{theorem}
\begin{proof}
For the number $m$ to be sufficiently large, 
the property $c<\lfloor \lg (m/4) \rfloor$ suffices,
as will be seen in the proof, where $\lg x$~is the logarithm of~$x$
to base~2. 
We partition an execution of algorithm~$\cA$ into $\lfloor m/8\rfloor$ 
stages, each lasting for $\lfloor \lg (m/4)\rfloor$ 
consecutive steps.
The adversary activates at least two, and up to four, stations holding
specific edges precisely at the start of each stage.
The chosen stations and edges are to have the following
property: none of these at least two edges is broadcast successfully 
during the stage when they are activated.

Consider the beginning of the $k$th~stage, 
where $k\le \lfloor m/8\rfloor$, and suppose that during 
the previous stages the execution has proceeded as required.
In particular, a total of at most $4(k-1)$ stations have been
activated, and none of the  edges activated 
$\lfloor \lg (m/4)\rfloor$ steps ago, when the
previous stage~$k-1$ started, has been heard by this step. 
The algorithm has not terminated yet because it is
$c$-correct, the inequality $c<\lfloor \lg (m/4) \rfloor$ holds and
the edges recently activated have to be in a spanning forest
of the input.

Prior to the beginning of a stage, we need a pool of passive
stations holding edges to choose the ones to activate from.
The specific fractions of~$m$ we use serve the purpose of the
proof for the following reason: each stage has at most four new edges
activated, for a total of $m/2$ edges after $m/8$ stages, 
and before each stage there are at least
$m-m/2=m/2$ edges to choose from.

Let the next consecutive queries be~$Q_i$, 
for $1\le i \le \lfloor \lg (m/4)\rfloor$.
They are determined uniquely if we assume that there are no
successful broadcasts of edges that have not been activated prior to
this point. 
This property will be guaranteed by the construction.
We proceed by considering sets $S_i$ and $E_i$, 
for $0\le i\le\lfloor\lg(m/4)\rfloor$.
Initialize the set~$S_0$ to the IDs of the stations still passive at
this point, and the set~$E_0$ to the edges held by these
stations.
The sets $S_i$ and $E_i$ are determined inductively.
Suppose $S_i$ and $E_i$ have already been determined, 
then $S_{i+1}$ and $E_{i+1}$ are defined as follows.
If $|Q_{i+1}\cap S_i| \ge |S_i|/2$, then set 
$S_{i+1} := Q_{i+1}\cap S_i$, otherwise set $S_{i+1} := Q_{i+1} - S_i$.
Similarly, if $|Q_{i+1}\cap E_i| \ge |E_i|/2$, then set 
$E_{i+1} := Q_{i+1}\cap E_i$, otherwise set $E_{i+1} := Q_{i+1} - E_i$.
The sets $S_i$ and $E_i$, for $i=\lfloor \lg (m/4)\rfloor$, 
contain at least two elements each.
Choose any two stations in the final set~$S_i$ and any two
stations in the final set~$E_i$ to activate in the next stage. 

When the last stage  has been completed, there are still 
at least two stations that have been activated prior to the last~$c$ 
steps and they hold edges that have to be in a spanning
forest of the graph with edges held by all the stations
active by this step. 
This means that the algorithm still needs to perform at
least two more steps.
\end{proof}

Theorem~\ref{t-adversarial-lower} can be strengthened to hold
for any graph that {\em contains\/} a forest of~$m$ edges, a proof
is a straightforward modification, details are omitted.

%Algorithm {\sc DetAdversarial}.

\begin{figure}

\rule{\textwidth}{0.5pt}

\begin{center}
\begin{minipage}{\pagewidth}
\addtolength{\baselineskip}{3pt}

{\sc Input:} the number $n$ of graph vertices ; edge $e_p$ ;

{\sc Initialization:}  
{\sc Revealed} $:=\emptyset$, {\sc Cycle} $:=\emptyset$, 
$e_p$ is waiting ;
\F

{\sf if} this is the start of an execution 

\T {\sf then} participate in {\sc Resolve} to elect a leader

\T {\sf else} {\sf begin} 

\TT wait for the first update message ;

\TT update {\sc Revealed} ; 

\TT adjust the status of $e_p$ 

\T {\sf end} ;

{\sf repeat}

\T {\sf if} $e_p$ is waiting {\sf then } broadcast a dummy message  ;

\T {\sf if} ( silence was heard in the previous step ) {\sf and}
( $p$ is a leader ) 

\TT {\sf then} broadcast a termination signal ;

\T $\{e\} := $ {\sc Resolve}\,(\,{\sc Waiting}\,) ;

\T move edge $e$ from {\sc Waiting} to {\sc Revealed} ;

\T {\sf if} ( $e_p$ is waiting ) {\sf and} 
		( there is a cycle in the graph induced by $e_p$ 

\T \TT \TT and the edges in {\sc Revealed} )

\TT {\sf then} move $e_p$ from {\sc Waiting} to {\sc Cycle} ; 

\T {\sf if} $p$ is a leader {\sf then} broadcast an update message ; 

{\sf until} a termination signal was heard

\F

{\sc Output:} all the revealed edges.

\end{minipage}
\end{center}

\rule{\textwidth}{0.5pt}

\centerline{\parbox{\captionwidth}{
\caption{
\label{fig-adversary-deter}
Algorithm {\sc DetAdversarial}.
Code for the station $p$ that holds the input edge~$e_p$.
} } }
\end{figure}

%%%%%%%%%%%%%%%%%%%%%%%%%%%%%%%%%%%%%%%%%%%%%%%%%%%%%%%%%%%%%%%%%%

Next we give a deterministic adaptive algorithm {\sc DetAdversarial},
whose performance matches the lower bound of 
Theorem~\ref{t-adversarial-lower}.
The design principle under which it operates is similar to that of
the adaptive algorithm for the static case.
The main difference is that now new stations can wake up at
arbitrary steps and they need to be incorporated into an
execution.
When a station is activated after the start of an execution,
it pauses and listens to the channel 
until it hears a special {\em update message}.
This message carries a list of all the edges that have
been revealed in the course of the execution so far.
Such update messages are sent by one designated station, called a
{\em leader}.
The algorithm starts by having those stations that are active from the
very beginning select a leader among themselves.
This is achieved by running the procedure {\sc Resolve}, 
and using IDs in it rather than edges.
If a station~$p$ joins the execution at some point, after having
been activated, then it waits for the first update message, 
then sets {\sc Revealed} to the list obtained, and if there is a 
cycle in the graph induced by $e_p$ and the edges in {\sc Revealed},
then $e_p$ becomes a cycle.
A pseudocode of the algorithm is presented in 
Figure~\ref{fig-adversary-deter}.

\begin{theorem}
Algorithm {\sc DetAdversarial} finds a spanning forest~$|T|$ 
in time $\cO(|T| \log m)$ and is $2$-correct against any adversary.
\end{theorem}

\begin{proof}
It is checked whether there is any waiting edge 
just before a possible termination signal is to be broadcast.
A station that has been activated at least two steps before has a
chance to receive the update message that was broadcast as the
last action performed in the previous iteration of the main loop.
It takes $|T|$~calls of the procedure {\sc Resolve} to
contribute all the edges, each call takes the
time~$\cO(\log m)$.
\end{proof}

%%%%%%%%%%%%%%%%%%%%%%%%%%%%%%%%%%%%%%%%%%%%%%%%%%%%%%%%%%%%%%%%%%%%

\section{Discussion}

This paper presents a study of the problem of finding
a minimum-weight spanning forest in a distributed setting, 
for the model when single edges are held by stations 
that communicate by broadcasting on a multiple-access channel.

We show that adaptive deterministic algorithms are more efficient
than oblivious ones, even for simple graphs without weights.
Finding the optimum performance of a deterministic adaptive 
algorithm for simple graphs is an open problem.
We claim that Theorem~\ref{thm-1} actually gives the best possible
bound.

We develop an optimal randomized algorithm for simple
graphs without weights.
It is an open problem if the performance of this algorithm can be
matched by that of a deterministic one.
We conjecture that this is not the case.

We also develop a randomized algorithm finding a  minimum-weight 
spanning forest~$T$ of a~graph in expected time
$\cO(|T|+ W \log m)$, where $W$ is the number of distinctive
weights on the edges of~$T$,
and show that any deterministic one requires time~$\Omega(m)$.
This shows that randomization helps for this problem, 
for sufficiently many edges.
The optimality of finding a minimum-weight spanning forest by a
randomized algorithm is an open problem. 

We develop a deterministic algorithm for an adversarial
environment that is time optimal, but its properties rely on
a possibly large size of a message broadcast in a single step.
An interesting problem is what is the optimum-time complexity
of the problem in an adversarial model with the size
of messages restricted so that each can carry up to a constant number 
of edges or IDs of stations?

%: bibliography

\bibliographystyle{plain}

\bibliography{bogdan,books,distributed,networks}

\begin{thebibliography}{10}

\bibitem{Abramson-TIT85}
Norman~M. Abramson.
\newblock Development of the {ALOHANET}.
\newblock {\em IEEE Transactions on Information Theory}, 31(2):119--123, 1985.

\bibitem{Awerbuch87-STOC}
Baruch Awerbuch.
\newblock Optimal distributed algorithms for minimum weight spanning tree,
  counting, leader election and related problems.
\newblock In {\em Proceedings of the $19$th Annual {ACM} Symposium on Theory of
  Computing (STOC)}, pages 230--240, 1987.

\bibitem{BertsekasG92}
Dmitri Bertsekas and Robert Gallager.
\newblock {\em Data Networks}.
\newblock Prentice Hall, second edition, 1992.

\bibitem{Chlebus-chapter-2001}
Bogdan~S. Chlebus.
\newblock Randomized communication in radio networks.
\newblock In Panos~M. Pardalos, Sanguthevar Rajasekaran, John~H. Reif, and Jose
  D.~P. Rolim, editors, {\em Handbook of Randomized Computing}, volume~I, pages
  401--456. Kluwer Academic Publishers, 2001.

\bibitem{ChlebusGK-TCS03}
Bogdan~S. Chlebus, K.~Go{\l}{\k a}b, and Dariusz~R. Kowalski.
\newblock Broadcasting spanning forests on a multiple-access channel.
\newblock {\em Theory of Computing Systems}, 36(6):711--733, 2003.

\bibitem{ChlebusGK-SIROCCO02}
Bogdan~S. Chlebus, Karol Go{\l}{\k a}b, and Dariusz~R. Kowalski.
\newblock Finding spanning forests by broadcasting.
\newblock In {\em Proceedings of the $9$th International Colloquium on
  Structural Information and Communication Complexity (SIROCCO)}, volume~13 of
  {\em Proceedings in Informatics}, pages 41--56. Carleton Scientific, 2002.

\bibitem{ChlebusKL-DC06}
Bogdan~S. Chlebus, Dariusz~R. Kowalski, and Andrzej Lingas.
\newblock Performing work in broadcast networks.
\newblock {\em Distributed Computing}, 18(6):435--451, 2006.

\bibitem{ClementiMS02}
Andrea E.~F. Clementi, Angelo Monti, and Riccardo Silvestri.
\newblock Optimal ${F}$-reliable protocols for the {Do-All} problem on
  single-hop wireless networks.
\newblock In {\em Proceedings of the $13$th International Symposium on
  Algorithms and Computation (ISAAC)}, volume 2518 of {\em Lecture Notes in
  Computer Science}, pages 320--331. Springer, 2002.

\bibitem{CormenLRS09-3rd-edition-book}
Thomas~H. Cormen, Charles~E. Leiserson, Ronald~L. Rivest, and Clifford Stein.
\newblock {\em Introduction to Algorithms}.
\newblock The MIT Press, third edition, 2009.

\bibitem{DworkHW98}
Cynthia Dwork, Joseph~Y. Halpern, and Orli Waarts.
\newblock Performing work efficiently in the presence of faults.
\newblock {\em SIAM Journal on Computing}, 27(5):1457--1491, 1998.

\bibitem{FaloutsosM95}
Michalis Faloutsos and Mart Molle.
\newblock Optimal distributed algorithm for minimum spanning trees revisited.
\newblock In {\em Proceedings of the Fourteenth Annual {ACM} Symposium on
  Principles of Distributed Computing (PODC)}, pages 231--237, 1995.

\bibitem{Gallager85}
Robert~G. Gallager.
\newblock A perspective on multiaccess channels.
\newblock {\em IEEE Transactions on Information Theory}, 31(2):124--142, 1985.

\bibitem{GallagerHS-TPLS83}
Robert~G. Gallager, Pierre~A. Humblet, and Philip~M. Spira.
\newblock A distributed algorithm for minimum-weight spanning trees.
\newblock {\em ACM Transactions on Programming Languages and Systems},
  5(1):66--77, 1983.

\bibitem{GarayKP-SICOMP98}
Juan~A. Garay, Shay Kutten, and David Peleg.
\newblock A sublinear time distributed algorithm for minimum-weight spanning
  trees.
\newblock {\em SIAM Journal on Computing}, 27(1):302--316, 1998.

\bibitem{GasieniecPP-JDM01}
Leszek G{\k a}sieniec, Andrzej Pelc, and David Peleg.
\newblock The wakeup problem in synchronous broadcast systems.
\newblock {\em SIAM Journal on Discrete Mathematics}, 14(2):207--222, 2001.

\bibitem{GoldbergMPS-JACM00}
Leslie~Ann Goldberg, Philip~D. MacKenzie, Mike Paterson, and Aravind
  Srinivasan.
\newblock Contention resolution with constant expected delay.
\newblock {\em Journal of the ACM}, 47(6):1048--1096, 2000.

\bibitem{GrableP97}
David~A. Grable and Alessandro Panconesi.
\newblock Nearly optimal distributed edge coloring in ${\cO}(\log \log n)$
  rounds.
\newblock {\em Random Structures and Algorithms}, 10(3):385--405, 1997.

\bibitem{GreenbergW-JACM85}
Albert~G. Greenberg and Shmuel Winograd.
\newblock A lower bound on the time needed in the worst case to resolve
  conflicts deterministically in multiple access channels.
\newblock {\em Journal of the ACM}, 32(3):589--596, 1985.

\bibitem{GrimmettS2001-book}
Geoffrey~R. Grimmett and David~R. Stirzaker.
\newblock {\em Probability and Random Processes}.
\newblock Oxford University Press, third edition, 2001.

\bibitem{HanckowiakKP99}
Michal Ha\'n\'ckowiak, Michal Karo\'nski, and Alessandro Panconesi.
\newblock A faster distributed algorithm for computing maximal matchings
  deterministically.
\newblock In {\em Proceedings of the Eighteenth Annual {ACM} Symposium on
  Principles of Distributed Computing (PODC)}, pages 219--228, 1999.

\bibitem{HastadLR-SICOMP96}
Johan H{\aa }stad, Frank~Thompson Leighton, and Brian Rogoff.
\newblock Analysis of backoff protocols for multiple access channels.
\newblock {\em SIAM Journal on Computing}, 25(4):740--774, 1996.

\bibitem{JurdzinskiKZ-PODC02}
Tomasz Jurdzi\'nski, Miroslaw Kuty{\l}owski, and Jan Zatopia\'nski.
\newblock Efficient algorithms for leader election in radio networks.
\newblock In {\em Proceedings of the $21$st ACM Symposium on Principles of
  Distributed Computing (PODC)}, pages 51--57, 2002.

\bibitem{JurdzinskiS02}
Tomasz Jurdzi\'nski and Grzegorz Stachowiak.
\newblock Probabilistic algorithms for the wakeup problem in single-hop radio
  networks.
\newblock In {\em Proceedings of the $13$th International Symposium on
  Algorithms and Computation (ISAAC)}, volume 2518 of {\em Lecture Notes in
  Computer Science}, pages 535--549. Springer, 2002.

\bibitem{KomlosG-TIT85}
J{\'a}nos Koml{\'o}s and Albert~G. Greenberg.
\newblock An asymptotically fast nonadaptive algorithm for conflict resolution
  in multiple-access channels.
\newblock {\em IEEE Transactions on Information Theory}, 31(2):302--306, 1985.

\bibitem{KushilevitzM-SICOMP98}
Eyal Kushilevitz and Yishay Mansour.
\newblock An ${\Omega}({D} \log ({N}/{D}))$ lower bound for broadcast in radio
  networks.
\newblock {\em SIAM Journal on Computing}, 27(3):702--712, 1998.

\bibitem{KuttenP98}
Shay Kutten and David Peleg.
\newblock Fast distributed construction of small $k$-dominating sets and
  applications.
\newblock {\em Journal of Algorithms}, 28(1):40--66, 1998.

\bibitem{Linial92}
Nathan Linial.
\newblock Locality in distributed graph algorithms.
\newblock {\em {SIAM} Journal on Computing}, 21(1):193--201, 1992.

\bibitem{LotkerPP06}
Zvi Lotker, Boaz Patt{-}Shamir, and David Peleg.
\newblock Distributed {MST} for constant diameter graphs.
\newblock {\em Distributed Computing}, 18(6):453--460, 2006.

\bibitem{Martel-IPL94}
Charles~U. Martel.
\newblock Maximum finding on a multiple access broadcast network.
\newblock {\em Information Processing Letters}, 52(1):7--15, 1994.

\bibitem{MartelV88}
Charles~U. Martel and Thomas~P. Vayda.
\newblock The complexity of selection resolution, conflict resolution and
  maximum finding on multiple access channels.
\newblock In {\em Proceedings of the $3$rd Aegean Workshop on Computing, {VLSI}
  Algorithms and Architectures (AWOC)}, volume 319 of {\em Lecture Notes in
  Computer Science}, pages 401--410. Springer, 1988.

\bibitem{MetcalfeB76}
Robert~M. Metcalfe and David~R. Boggs.
\newblock Ethernet: {D}istributed packet switching for local computer networks.
\newblock {\em Communications of the ACM}, 19(7):395--404, 1976.

\bibitem{Peleg2000-book}
David Peleg.
\newblock {\em Distributed Computing: {A} Locality-Sensitive Approach}.
\newblock Society for Industrial and Applied Mathematics, 2000.

\bibitem{PelegR00}
David Peleg and Vitaly Rubinovich.
\newblock A near-tight lower bound on the time complexity of distributed
  minimum-weight spanning tree construction.
\newblock {\em {SIAM} Journal on Computing}, 30(5):1427--1442, 2000.

\bibitem{RaghavanU-SICOMP98}
Prabhakar Raghavan and Eli Upfal.
\newblock Stochastic contention resolution with short delays.
\newblock {\em SIAM Journal on Computing}, 28(2):709--719, 1998.

\bibitem{TanenbaumW10}
Andrew~S. Tanenbaum and David Wetherall.
\newblock {\em Computer Networks}.
\newblock Prentice Hall, fifth edition, 2010.

\bibitem{Willard-SICOMP86}
Dan~E. Willard.
\newblock Log-logarithmic selection resolution protocols in a multiple access
  channel.
\newblock {\em SIAM Journal on Computing}, 15(2):468--477, 1986.

\end{thebibliography}

\end{document}